\documentclass[11pt,a4paper]{article}

\usepackage[margin=1in]{geometry}
\usepackage{microtype}

\usepackage{times}

\usepackage[square,numbers]{natbib}
\usepackage{xcolor} 
\usepackage{amsmath,amssymb} 
\usepackage{amsthm} 
\usepackage{amsfonts} 
\usepackage{nicefrac} 
\usepackage{subfigure} 
\usepackage{enumitem} 
\usepackage{setspace} 
\usepackage{algorithm2e} 
\usepackage{hyperref} 
\usepackage{graphicx}
\usepackage{physics} 
\usepackage{cleveref} 
\usepackage{multirow}
\usepackage[most]{tcolorbox}

\usepackage{booktabs}
\usepackage{tabularx}

\usepackage{tikz} 
\usetikzlibrary{arrows.meta} 
\usetikzlibrary{decorations.markings}
\usetikzlibrary{backgrounds}
\usetikzlibrary{positioning,chains,fit,shapes,calc}
\usetikzlibrary{angles,patterns}

\usetikzlibrary{matrix,quotes}
\usetikzlibrary{decorations.pathmorphing, calc}

\newtheorem{proposition}{Proposition}[section]
\newtheorem{lemma}{Lemma}[section]
\newtheorem{remark}{Remark}[section]

\newtheorem{definition}{Definition}[section] 

\newcommand{\Z}{\mathbb{Z}}

\newcommand{\R}{\mathbb{R}}

\newcommand{\T}{\mathbb{T}}

\newcommand{\vc}[1]{\mathbf{#1}} 

\newcommand{\inv}[1]{#1^{-1}}

\newcommand{\floor}[1]{\left\lfloor{#1}\right\rfloor} 
\newcommand{\innerprod}[1]{\langle{#1}\rangle} 

\usepackage{tikz}
\usetikzlibrary{shapes, arrows}
\tikzstyle{block} = [rectangle, rounded corners, draw=black, thick,
    text centered, minimum height=2em, minimum width=8em]
\tikzstyle{blueblock} = [block, fill=blue!20]
\tikzstyle{redblock} = [block, fill=red!20]
\tikzstyle{line} = [draw, thick, -latex']

\title{Has MIMO decoding been proved hard from lattice problems?} 

\author{
  Yang~Li
  School of Information Technology\\
  Deakin University\\
  Burwood, VIC 3125 \\
  \texttt{kelvin.li@deakin.edu.au} \\
}

\begin{document}
\maketitle


\begin{abstract}
Multiple-input multiple-output (MIMO) technology is fundamental to modern wireless communication. Physical layer security seeks to protect transmitted information by exploiting properties of the noisy communication channel. Dean and Goldsmith proposed a polynomial time reduction from lattice problems to MIMO decoding by adapting Regev's reduction for learning with errors (LWE). If valid, this reduction would give physical layer security a strong computational foundation based on the hardness of established lattice problems. Subsequent works presented attacks and counterexamples against the resulting construction, casting doubt on its security but leaving the precise validity and limitations of the underlying reduction incompletely understood. We provide a theoretical examination of the revised reduction and identify the structural features of the LWE reduction that fail to carry over to the non-modular MIMO setting, hence showing that its published proof does not establish the claimed hardness of MIMO decoding. Our results distinguish flaws in the hardness proof from direct attacks on particular parameter choices and clarify what would be required of any attempted repair. We do not rule out physical layer security for MIMO systems in general, but show that the claimed lattice hardness guarantee does not follow from the existing reduction.
\end{abstract}

\section{Introduction}

Multiple-input and multiple-output (MIMO) transmitting systems have been a core technology in wireless communication, due to their ability to accelerate data transmission by using multiple antennas on both the transmitter and receiver sides. As with other transmission systems, the security of MIMO communication has been extensively studied and remains critical to its practical deployment. Some research efforts have explored the physical layer security of MIMO systems, aiming to achieve secure data transmission based solely on the system's physical properties such as transmission noise, without relying on software layer security mechanisms. 

Along this line of work, Dean and Goldsmith observed the resemblance between the MIMO decoding problem and the learning with errors (LWE) problem \cite{regev2009lattices} in their original and a revised papers \cite{dean2013physical,dean2017physical}. In their work, Dean and Goldsmith attempted a similar hardness proof to Regev's LWE reduction, by constructing a polynomial time hybrid (classical and quantum) reduction from hard lattice problems to MIMO decoding. 
However, these attempts have been shown to be flawed by counterexamples presented independently by multiple groups \cite{steinfeld2015massive, korzhik2017investigation, sakzad2020comments}. 

In this work, we take a closer theoretical look at the nature of MIMO decoding and the LWE reduction strategies, and analyse why such reduction techniques from hard lattice problems do not carry over to the MIMO context. Consequently, it remains necessary to carefully determine whether the LWE like reduction produces samples from the distributions required by MIMO oracle calls. This proof level question is the main focus of our work.
Our contributions are summarised as follows.
\begin{itemize}
    \item We provide a systematic examination of the revised Dean-Goldsmith reduction \cite{dean2017physical} and identify several steps for which the required distributional or parameter claims are not established.
    
    \item We also show that some immediate modifications either leave parts of the reduction unsuccessful or unresolved, or change the underlying MIMO problem into a modular or restricted variant.
    
    \item Our analysis indicates that an LWE style reduction cannot be transferred directly to real-valued, non-modular noisy linear systems without introducing additional structure that reproduces the modular cancellation and conditional randomness used in Regev's proof. We illustrate this limitation through a subsequent MIMO construction whose related lattice reduction exhibits the same proof-level concerns.
\end{itemize}

We emphasise that our conclusions concern the claimed reduction rather than the general possibility of MIMO physical layer security. Showing that a published reduction does not establish its theorem does not, by itself, constitute a decoding attack, nor does it prove that every MIMO based security construction is impossible.

The remainder of this paper is organised as follows. Section 2 introduces the necessary lattice, LWE and MIMO background specific to the Dean-Goldsmith construction. 
Sections 3 to 5 examine the principal Dean-Goldsmith reduction steps and their gaps. Section 6 extends the examination of the BDD-to-MIMO reduction to a subsequent construction. Section 7 concludes the paper.

\section{Preliminary}


\subsection{Lattice}

The works in \cite{dean2013physical,dean2017physical} do not involve much of lattice theory, hence we will cover only the relevant concepts in this section. 

\begin{definition}
Let $\vc{v_1}, \dots, \vc{v_n} \in \R^m$ be a set of linearly independent vectors. The \textbf{lattice} \index{lattice} $L$
generated by $\vc{v_1}, \dots, \vc{v_n}$ is the set of integer linear combinations of them. That is, 
\begin{equation*}
    L = \{a_1 \vc{v_1} + \cdots + a_n \vc{v_n} \mid a_1, \dots, a_n \in \Z\}.
\end{equation*}
\end{definition}

If the set $B = \{\vc{v_1}, \dots, \vc{v_n}\}$ is linearly independent and generates the lattice, then $\vc{B}$ is a basis of the lattice, denoted as $L(\vc{B})$. Similar to basis of a vector space, a lattice basis need not be unique. Different to vector space basis, a lattice need not have an orthogonal basis. Hence, some lattice bases are considered ``better'' than others, if their vectors are more orthogonal and short. For example, the reductions in \cite{regev2009lattices,dean2013physical,dean2017physical} applied a lattice basis reduction algorithm, LLL \cite{lenstra1982factoring}, for the purpose of reducing the basis vectors' norms with an upper bound. 
 


Given lattices are discrete spaces, some interesting computational problems arise in lattices. These problems are interesting from mathematical perspective, as well as from cryptographic perspective because some of them became the security foundations of post quantum cryptography. The most well known lattice problems are the following two. 
\begin{definition}[Shortest Vector Problem (SVP)]
Given a lattice basis $\vc{B}$, find a shortest non-zero vector in the lattice $L(\vc{B})$, i.e., find a non-zero vector $\vc{v} \in L(\vc{B})$ such that $||\vc{v}|| = \lambda_1(L(\vc{B}))$. 
\end{definition}

\begin{definition}[Closest Vector Problem (CVP)]
Given a lattice basis $\vc{B}$ and a target vector $\vc{t}$ that is not in the lattice $L(\vc{B})$, find a vector in $L(\vc{B})$ that is closest to $\vc{t}$. 
\end{definition}


A special case of CVP is the bounded distance decoding (BDD) problem, which is used in LWE and MIMO reductions. The difference being a BDD target is within a bounded distance to the lattice. 
\begin{definition}[The Bounded Distance Decoding Problem]
Given a lattice basis $\vc{B}$ of an $n$-dimensional lattice $L$ and a target vector $\vc{t} \in \R^n$ satisfies $dist(\vc{t},L) \le \lambda_1(L)$, find a lattice vector $\vc{v} \in L$ that is closest to $\vc{t}$. 
\end{definition}


\subsection{Learning with errors}
\label{sec:lwe}

Given the MIMO reduction closely follows the learning with errors (LWE) reduction by \citeauthor{regev2009lattices}, we cover in this section the LWE distribution, LWE problem, and outline its reduction strategy on a very high level. 

\begin{definition}
\label{def:lweDist}
For a given vector $\vc{s} \in \Z_q^n$, the \textbf{LWE distribution} $\vc{A}_{\vc{s},\chi}$ over $\Z_q^n \times \Z_q$ is sampled by 
\begin{itemize}
    \item generating a random vector $\vc{a} \leftarrow \Z_q^n$,
    \item generating a random noise $\epsilon \xleftarrow{\chi}\; \Z_q$,
    \item outputting $(\vc{a}, b = \innerprod{\vc{s}, \vc{a}} + \epsilon \bmod q)$.
\end{itemize}  
\end{definition}

The LWE distribution was defined in modular integer domain, which is one of the key reasons behind its hardness reduction from lattice problems. A compact way to express multiple LWE samples is through a matrix form, so that we have $(\vc{A}, \vc{b}) \in \Z_q^{m \times n} \times \Z_q$ representing $m$ LWE samples.

\begin{definition}
For the parameter $q$ and the error distribution $\chi$ over $\Z_q$, the \textbf{LWE problem}, denoted by LWE$_{q,\chi}$, is to find the unknown vector $\vc{s}$ given LWE samples  $(\vc{A},\vc{b}) \sim \vc{A}_{\vc{s},\chi}$. 
\end{definition}


The LWE work by \citeauthor{regev2009lattices} established a polynomial time hybrid (classic and quantum) reduction from lattice problems to LWE. It is one of the key founding theories in lattice based cryptography \cite{regev2009lattices}. Here, we give a brief overview of the reduction, to provide sufficient context so that the reader can connect the MIMO reduction to the LWE reduction. 

There are three main steps in the LWE reduction as follows, also shown in \Cref{fig:lweReduction} and \Cref{fig:lweReductionIterativeStep}.
\begin{enumerate}
    
    \item A CVP oracle can be built from an LWE oracle with large discrete Gaussian samples.

    \item The CVP oracle with a quantum routine can produce smaller discrete Gaussian samples.

    \item The lattice problems can be solved efficiently with sufficiently small discrete Gaussian samples.
\end{enumerate}
Steps 1 and 2 are repeated till the discrete Gaussian samples are sufficiently small to solve the lattice problems. In \cite{dean2013physical,dean2017physical}, Steps 2 and 3 from \cite{regev2009lattices} are used almost like black boxes, hence we only explain Step 1 to provide the context for MIMO reduction.  

To build a CVP oracle upon an LWE oracle, the key step is to embed the CVP solution $\kappa_{L^*}\vc{x}$ as the secret key for LWE samples, which can then be supplied to the LWE oracle to find $\kappa_{L^*}\vc{x}$. This is achieved by producing a sample  
\begin{align*}
    (L^{-1}\vc{v} \bmod p, \innerprod{\vc{x},\vc{v}}/p+ e \bmod 1).
\end{align*}
The first part of the above sample is proved to be uniform in $\Z_p^n$. The second part can be decomposed into
\begin{align*}
    \innerprod{\vc{x}, \vc{v}}/p + e \bmod 1 &= \innerprod{\kappa_{L^*}\vc{x} + \vc{\delta}, \vc{v}}/p + e \bmod 1\\
    &= \underbrace{\innerprod{\kappa_{L^*}\vc{x}, \vc{v}}/p}_{\text{embedded secret}} + \underbrace{(\innerprod{\vc{\delta}, \vc{v}}/p + e)}_{\text{aggregated  noise}} \bmod 1
\end{align*}
Hence, the sample is almost identical to an LWE sample, provided the parameters are within the correct ranges. This reveals two important substeps within the iterative structure in \Cref{fig:lweReductionIterativeStep}. 
\begin{itemize}
    \item First, the LWE oracle needs to be able to work with smaller errors, because the aggregate error above may not match the exact distribution of $e \sim \nu_{\alpha}$ in the LWE distribution definition. The solution is to pad samples with additional noise (Lemma 3.7 \cite{regev2009lattices}), while verifying the padded samples  match the LWE distribution (Lemma 3.6 \cite{regev2009lattices}). 
    
    \item Second, the immediately constructed CVP oracle from the LWE oracle and the discrete Gaussian samples only outputs the solution coefficients modular $p$. This is also to guarantee the sample follows closely with the LWE distribution for both parts (Corollary 3.10 and Lemma 3.11 \cite{regev2009lattices}).

\end{itemize}





\begin{figure}[hbt!]
    \centering
    \includegraphics[page=14]{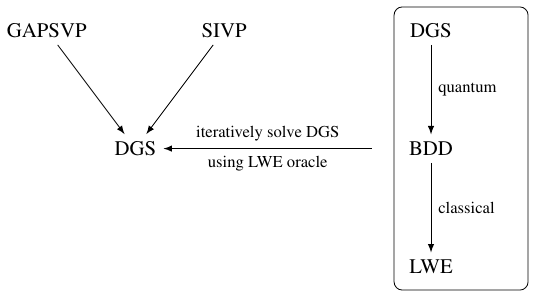}
    \caption{Reductions from the lattice problems to the LWE problem.}
    \label{fig:lweReduction}
\end{figure}

\begin{figure}[hbt!]
    \centering
    \includegraphics[page=24]{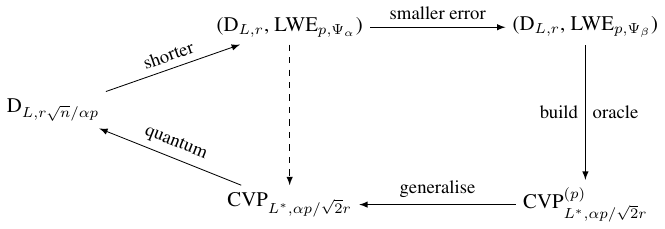}
    \caption{The iterative process in the LWE reduction. It starts with an assumed  LWE$_{p,\Psi_{\alpha}}$ oracle and discrete Gaussian samples. The process stops when the reproduced discrete Gaussian samples are short enough to solve the lattice problems.}
    \label{fig:lweReductionIterativeStep}
\end{figure}


\subsection{MIMO decoding}
Consider a real-valued MIMO system that consists of $n$ transmit antennas and $m$ receive antennas. Let $\vc{A} \in \R^{m \times n}$ be the channel matrix that affects the input signal before it reaches the receiver. The MIMO channel model is captured by
\begin{align*}
    \vc{y}=\vc{A} \vc{x} + \vc{e}.
\end{align*}
In \cite{dean2013physical,dean2017physical}, the entries of $\vc{A}$ are i.i.d. continuous Gaussian samples from $\nu_k := N(0, k^2/2\pi)$, where the parameter $k$ is known as scale in lattice based cryptography. 
The noise vector $\vc{e}$ also has i.i.d. entries but from a different Gaussian distribution $\nu_{\alpha} := N(0,\alpha^2/2\pi)$.
The input signal $\vc{x}$'s entries are taken from an uncentred $M$-PAM constellation $\Omega=[0,M) \cap \Z$ for a predetermined integer $M$.

To decode the MIMO system, the receiver could compute its Moore-Penrose inverse\footnote{Since $m \ge n$ and $\vc{A}$'s entries are independently sampled from a continuous Gaussian distribution, the channel matrix $\vc{A}$ has full column rank with probability one, hence the particular simple algebra expression of the pseudoinverse.} $\vc{A}^{\dagger} = \inv{(\vc{A}^T \vc{A})}\vc{A}^T$ to get
\begin{align*}
    \vc{A}^{\dagger} \vc{y} = \vc{x} + \vc{A}^{\dagger}\vc{e}.
\end{align*}
This is also known as zero-forcing in signal processing. A limitation of zero-forcing is that if $\vc{A}$ has a small minimum singular value, $\vc{A}$'s pseudoinverse could amplify the noise magnitude, increasing the probability of decoding error.

Given the signal is transmitted on a public channel, decoding is not prohibited for any eavesdropper on the channel. Dean and Goldsmith then used SVD precoding to create a decoding asymmetry, which they claimed would make decoding exponentially hard for the eavesdropper. 

To do so, before transmitting the signal, the sender decomposes the channel matrix to its singular values $\vc{A} = \vc{U} \vc{\Sigma} \vc{V}^T$, and transmits the preprocessed signal 
\begin{align*}
    \Tilde{\vc{x}} = \vc{V} \vc{x}
\end{align*}
through the MIMO channels instead of the original signal $\vc{x}$. This process is known as linear precoding.

Upon receiving the signal in the following form
\begin{align}
    \label{eq:receiver}
        \vc{y} = A \Tilde{\vc{x}} + \vc{e} =  \vc{U} \vc{\Sigma} \vc{V}^T \vc{V} \vc{x} + \vc{e} = \vc{U} \vc{\Sigma} \vc{x} + \vc{e},
    \end{align}
the receiver postprocesses it to obtain 
    \begin{align*}
        \Tilde{\vc{y}} = \vc{U}^T \vc{y} = \vc{U}^T \vc{U} \vc{\Sigma} \vc{x} + \vc{U}^T \vc{e} = \vc{\Sigma} \vc{x} + \Tilde{\vc{e}}.
    \end{align*}
As long as the noise is reasonably small, the receiver can decode $\Tilde{\vc{y}}$ to accurately approximate the original signal $\vc{x}$.

For the eavesdropper, their wiretap channel has a different channel matrix $\vc{B}$, which then transforms their received signal in the form
\begin{align*}
    \Bar{y} = \vc{B}\vc{V}\vc{x} + \Bar{e}
\end{align*}
that is different from what the legitimate receiver obtained in \Cref{eq:receiver}. Given $\vc{V}$ is part of $\vc{A}$'s SVD, it is unitary hence implies $\vc{B} \vc{V}$ is still a matrix with i.i.d. Gaussian entries. In this case, the eavesdropper gets no hint on decoding the received signal, hence the decoding task is no different from decoding a MIMO transmitted signal with no preprocessing step by the sender. 




\subsection{MIMO distribution}

\begin{definition}
For an arbitrary signal vector $\vc{x} \leftarrow \Omega^n$, the \textbf{MIMO distribution} $\vc{A}_{M,\alpha,k}$ over $\R^n \times \R$ is obtained by taking the following steps. At the $i$th receive antenna, with $M, \alpha, k$ being set:
    \begin{enumerate}
        
        \item Sample $\vc{a}_{i} \leftarrow \Psi_k$ over $\R^n$, where $\vc{a}_i$ is the $i$th row of $\vc{A}$.

        \item Sample $e_i \leftarrow \Psi_{M \alpha}$ over $\R$.

        \item Compute $y_i = \innerprod{\vc{a}_i, \vc{x}} + e_i$.

        \item Output $(\vc{a}_i, y_i)$.
    \end{enumerate}
\end{definition}

\paragraph{An equivalent definition} Given that the signal constellation $\Omega$ and the noise distribution $\Psi_{M\alpha}$ have a common integer factor $M$, the MIMO distribution can be equivalently defined as 
\begin{align*}
    \vc{A}_{M,\alpha,k} = \{\vc{a}_i, y_i = \innerprod{\vc{a}_i, \vc{x}}/M + e_i\}_{i=1}^N, \text{ where } e_i \leftarrow \Psi_{\alpha}.
\end{align*}

\begin{definition}[MIMO-Search$_{M,\alpha,k}$]
    Let $M \ge 2$, $\alpha \in (0,1)$, $k \in \R$, $n > 0$. Given a polynomial number of samples from the MIMO distribution $\vc{A}_{M,\alpha,k}$ that forms an $m \times n$ matrix $\vc{A}$, find the signal vector $\vc{x} \in \Omega^n$.
\end{definition}

\paragraph{MIMO decoding is equivalent to solving MIMO-search} According to the above definition, the MIMO-Search problem is equivalent to the MIMO decoding problem. More precisely, drawing $m$ samples from the MIMO distribution $\vc{A}_{M,\alpha,k}$ produces $(\vc{A}, \vc{y}) \in \R^{m\times n} \times \R^m$, the first part of the MIMO samples corresponds to the channel matrix $\vc{A}$, and the second part of the MIMO samples corresponds to the received signal $\vc{y}$. Hence, finding the vector $\vc{x}$ in MIMO-search is equivalent to the MIMO decoding problem. 

The reduction of the lattice problems to LWE utilises much of its discreteness structure, \citeauthor{dean2017physical} hence also introduced a discrete analogue of the (continuous) MIMO distribution.\footnote{The definition informally appears before Lemma 1 in Section B \cite{dean2017physical}.} 

The discrete MIMO definition is based on the discrete Gaussian distribution over a lattice $L$. To recall, for all lattice vectors $\vc{x} \in L$, the discrete Gaussian distribution is defined as 
\begin{align*}
    D_{L,s,\vc{c}}(\vc{x}) = \frac{\rho_{s,\vc{c}}(\vc{x})}{\rho_{s,\vc{c}}(L)},
\end{align*}
where $\rho_s(\vb*{x}) = \exp( -\norm{\vb*{x}}^2/2\sigma^2)$ is the $0$-centred Gaussian function.

\begin{definition}
    Given an arbitrary lattice $L(\vc{A})$ and a number $r > \sqrt{2} \eta_{\epsilon}(L(\vc{A}))$, sample a vector $\vc{a} \leftarrow D_{L(\vc{A}),r}$, and an error $e \leftarrow \Psi_{\alpha}$, then output a discrete sample 
\begin{align*}
    \left(k \vc{a}/r, y = \innerprod{k \vc{a} / r, \vc{x}}/M + e \right) \leftarrow D_{M,\alpha,k}
\end{align*}
from the \textbf{discrete MIMO distribution}. 
\end{definition}

    



\subsection{Summary of similar distributions}

To complete this section, and for the ease of reference, we create a table of notations, and a table of similar LWE like distributions that appear throughout the paper.



\begin{table}[htbp]
\centering
\caption{Notation table}
\label{tab:notation}

\renewcommand{\arraystretch}{1.2}
\begin{tabularx}{0.8\linewidth}{@{} l l X @{}}
\toprule
\textbf{Notation} & \textbf{Meaning} \\
\midrule

$\nu_s$
    & Gaussian distribution $N(0,s^2/2\pi)$\\

$\nu_s^n$
    & $n$-dimensional $\nu_s$\\

$\T=\R/\Z$
    & 1-dimensional torus $\R \bmod \Z$. \\

$\Omega=[0,M)\cap\Z$
    & Uncentred $M$-PAM signal constellation for $M\geq 2$. \\

\bottomrule
\end{tabularx}
\end{table}

\begin{table}[htbp]
\centering
\caption{Comparison of the distributions considered in this paper.}
\label{tab:distribution-comparison}
\small
\renewcommand{\arraystretch}{1.25}
\begin{tabularx}{\linewidth}{
@{}
l
>{\raggedright\arraybackslash}p{0.27\linewidth}
>{\raggedright\arraybackslash}X
@{}
}
\toprule
Distribution & Domain & Sampling procedure \\
\midrule
LWE
&
$\vc{s}\in\Z_p^n$, $p\geq2$;
$\vc{A}_{\vc{s},\phi}$ is over $\Z_p^n\times\T$
&
For each $i\in[N]$, sample $\vc{a}_i$ uniformly from
$\Z_p^n$ and $e_i\leftarrow\phi$, and output
$\left(\vc{a}_i,
y_i=\innerprod{\vc{a}_i,\vc{s}}/p+e_i\bmod1\right)$
\\
\addlinespace
CLWE
&
$\vc{s}\in\R^n$, $\norm{\vc{s}}=1$, and $\gamma>0$;
$\vc{A}_{\vc{s},\beta,\gamma}$ is over $\R^n\times\T$
&
For each $i\in[N]$, sample
$\vc{a}_i\leftarrow\nu_1^n$ and
$e_i\leftarrow\nu_\beta$, and output
$\left(\vc{a}_i,
y_i=\gamma\innerprod{\vc{a}_i,\vc{s}}+e_i\bmod1\right)$
\\
\addlinespace
ILWE
&
$\vc{s}\in\Z^n$;
$\mathcal{D}_{\vc{s},\chi_a,\chi_e}$ is over
$\Z^n\times\Z$
&
For each $i\in[N]$, sample
$\vc{a}_i\leftarrow\chi_a^n$ and
$e_i\leftarrow\chi_e$, and output
$\left(\vc{a}_i,
y_i=\innerprod{\vc{a}_i,\vc{s}}+e_i\right)$
\\
\addlinespace
MIMO
&
$\vc{s}\in\Omega^n$;
$\vc{A}_{M,\alpha,k}$ is over $\R^n\times\R$
&
For each $i\in[N]$, sample
$\vc{a}_i\leftarrow\nu_k^n$ and
$e_i\leftarrow\nu_\alpha$, and output
$\left(\vc{a}_i,
y_i=\innerprod{\vc{a}_i,\vc{s}}/M+e_i\right)$
\\
\addlinespace
Discrete-MIMO
&
$\vc{s}\in\Omega^n$;
$D_{M,\alpha,k}$ is over $(k/r)L\times\R$
&
For each $i\in[N]$, sample
$\vc{u}_i\leftarrow D_{L,r}$ and
$e_i\leftarrow\nu_\alpha$, define
$\vc{a}_i=(k/r)\vc{u}_i$, and output
$\left(\vc{a}_i,
y_i=\innerprod{\vc{a}_i,\vc{s}}/M+e_i\right)$
\\
\bottomrule
\end{tabularx}
\end{table}

\section{The gap in discrete-to-continuous-MIMO reduction}

The Dean-Goldsmith reduction introduces an auxiliary discrete-MIMO distribution supported on a scaled lattice, while the assumed MIMO oracle expects samples with continuous Gaussian public components. Lemma 1 of \cite{dean2017physical} is intended to bridge these two domains. Its argument relies on the observation that, in the standard computational model, algorithms and oracles receive finite bit strings, so real-valued inputs must be represented using polynomially many bits. This finite precision assumption is acceptable in principle. However, geometric density at the chosen precision does not establish that the induced finite precision distribution is statistically close to the quantised continuous-MIMO distribution. 

This and the following two sections start with the original Lemma statement in \cite{dean2017physical}, followed by a brief explanation of their proof strategy, a theoretical examination, and remarks.

\paragraph{Lemma 1 \cite{dean2017physical}} Continuous-to-Discrete Samples. Given an oracle which can solve MIMO-Decision$_{M,\alpha,k}$, there exists an efficient algorithm to recover $\vc{x}$ given samples from $D_{M,\alpha,k}$.\footnote{There is a typo in MIMO-Decision. It should be MIMO-Search.}

\paragraph{Proof strategy} The key step to this lemma is densifying the lattice support $(k/r)L \times \R$ of the discrete-MIMO distribution, so that it is statistically indistinguishable from the continuous support $\R^n \times \R$ to the MIMO-search oracle. To do so, Dean and Goldsmith introduced the linear combinations of two discrete-MIMO samples 
\begin{align*}
    \lambda_i(\vc{a}_i, y_i) + \lambda_j (\vc{a}_j, y_j) &= \left(
    \lambda_i\vc{a}_i+\lambda_j\vc{a}_j,\,
    \lambda_i\vc{y}_i+\lambda_j\vc{y}_j
\right)\\
&=
\left(
    \lambda_i\vc{a}_i+\lambda_j\vc{a}_j,\,
    \left\langle
        \lambda_i\vc{a}_i+\lambda_j\vc{a}_j,
        \vc{x}
    \right\rangle
    +\lambda_i e_i+\lambda_j e_j
\right),
\end{align*}
where the coefficients are
\begin{align*}
\lambda_i=\frac{c_i}{c_i+c_j},
\qquad
\lambda_j=\frac{c_j}{c_i+c_j},
\qquad
\lambda_i+\lambda_j=1,
\end{align*}
calculated based on two arbitrary integers $c_i, c_j \in \Z_{2^{n^c}}$ generated according to the practical precision limit of the MIMO system. 

\begin{remark}
There are two issues in the above construction. First, the combined noise satisfies
\begin{align*}
\lambda_i e_i+\lambda_j e_j
\sim
N\left(
    0,
    \frac{(\lambda_i^2+\lambda_j^2)\alpha^2}{2\pi}
\right),
\end{align*}
which is generally narrower than $\nu_\alpha$. This variance defect can be removed algebraically by using coefficients whose squares sum to one, such as
\begin{align*}
\sqrt{\lambda_i}(\vc{a}_i,y_i)
+
\sqrt{\lambda_j}(\vc{a}_j,y_j).
\end{align*}
This normalisation preserves the covariance of both the public and noise components. However, it does not establish the required distribution. Given $\vc{a}_i$ and $\vc{a}_j$ are sampled from a discrete Gaussian over a lattice $(k/r)L$, then
\begin{align*}
\sqrt{\lambda_i}\vc{a}_i
+
\sqrt{\lambda_j}\vc{a}_j
\end{align*}
is supported on
\begin{align*}
\sqrt{\lambda_i}(k/r)L
+
\sqrt{\lambda_j}(k/r)L.
\end{align*}
Its distribution is induced by taking the weighted sum of two independent discrete Gaussian samples. This support need not be a lattice, so preserving the variance does not resolve the required statistical closeness argument.

The second issue is the missing proof that the resulting sample distribution is statistically close to the continuous MIMO distribution support $\R^n \times \R$. Dean and Goldsmith attempts to argue the closeness by leveraging Claim 3.9 in \cite{regev2009lattices}. They claim that the discrete Gaussian distribution $D_{L, r}$ is within a negligible distance to $\nu_r$, given $r$ is above the smoothing parameter $\eta(L)$ is 
actually a false claim. Regev's Claim 3.9 proves that when adding a discrete Gaussian sample with a continuous Gaussian noise, the resulting sample is statistically close to another continuous Gaussian sample, provided both the discrete and continuous Gaussian distributions are sufficiently wide. 

Therefore, the published proof does not establish Lemma 1. This observation does not rule out a different discrete-to-continuous transformation, possibly using Gaussian convolution under additional smoothing conditions. We do not pursue such a repair here. Instead, in the following sections we grant Lemma 1 and examine whether the remaining reduction follows even under that assumption.
\end{remark}


\section{The gap in verifying discrete-MIMO solution}

Similar to the LWE reduction described in \Cref{sec:lwe}, the transformation of a lattice decoding instance may produce MIMO samples whose effective Gaussian noise has some unknown scale $\beta\leq\alpha$, rather than the scale $\alpha$ expected by the MIMO oracle. The error handling procedure attempts to pad this effective noise with an independent Gaussian, so that the resulting samples have the required noise distribution. Because the original scale $\beta$ is unknown, the procedure tries a polynomial number of padding parameters and requires an efficient method for determining whether the oracle has returned the correct secret. Lemma 2 of \cite{dean2017physical} is intended to provide this verification procedure.

\paragraph{Lemma 2 \cite{dean2017physical}} Verifying solutions of MIMO-Search$_{M,\alpha,k}$. There exists an efficient algorithm that, given $\vc{x}'$ and a polynomial number of samples from $\vc{A}_{\vc{x},\alpha,k}$, for an unknown $\vc{x}$, outputs whether $\vc{x}=\vc{x}'$ with overwhelming probability.

\paragraph{Proof strategy} To build the verifier for the continuous MIMO search problem, Dean and Goldsmith's construction is taking the difference $y - \innerprod{\vc{a}, \vc{x}'}$, where $y$ is part of a genuine MIMO sample $(\vc{a}, y)$. Then distinguish the two distributions to verify whether or not $\vc{x} = \vc{x}'$. This follows the same proof strategy as its counterpart (Lemma 3.6 \cite{regev2009lattices}) in the LWE reduction. The difference being the resulting sample in Dean and Goldsmith's construction is still a $0$ centred Gaussian distribution, but with a wider variance than the original noise distribution $\nu_{\alpha}$. However, in Regev's LWE construction, the resulting sample distribution has a smaller period $1/k$ than the original wrapped noise distribution on the torus $\T$. 

\begin{lemma}[A corrected verification lemma]
Let $\vc{x},\vc{x}' \in \{0,\ldots,M-1\}^{n}$, and suppose that
\begin{align*}
    \frac{k}{M\alpha} \geq \frac{1}{p(n)}
\end{align*}
for some polynomial $p$. There exists a probabilistic polynomial-time
algorithm that, given $\vc{x}'$ and polynomially many independent
samples from $\vc{A}_{\vc{x},\alpha,k}$ for an unknown $\vc{x}$,
determines whether $\vc{x}=\vc{x}'$ with overwhelming
probability.
\end{lemma}

The correction fixes two issues in Dean and Goldsmith's Lemma 2 proof. First, the missing divisor $M$ in the MIMO distribution sample. Second, since Dean and Goldsmith's approach is based on the distribution being continuous, and looking at the difference between the two sample variances, the best practice is to look at the relative difference between the two variances, not the absolute difference. 

\begin{proof}
The construction leads to one of the two outcomes below.
\begin{align*}
y - \innerprod{\vc{a}, \vc{x}'}/M = e + \innerprod{\vc{a}, \vc{x} - \vc{x}'}/M=
\begin{cases}
    e, &
    \text{if } \vc{x}=\vc{x}',\\
    e+ \innerprod{\vc{a},\vc{x}-\vc{x}'}/M, &
    \text{if } \vc{x}\neq\vc{x}'.
\end{cases}
\end{align*}
Given $\vc{a} \sim \nu_k^n$ is a vector of i.i.d. Gaussian samples from $N(0, k^2/2\pi)$, and $\vc{d} = \vc{x} - \vc{x}'$ is a constant vector, their inner product is a linear combination of i.i.d. Gaussians that follows
\begin{align*}
    \innerprod{\vc{a}, \vc{d}}/M \sim N(0, \frac{\norm{d}^2 k^2}{M^2 2\pi}).
\end{align*}
Adding it to the independent noise then produces an even wider Gaussian distribution
\begin{align*}
    e + \innerprod{\vc{a}, \vc{d}}/M \sim N(0, V_1 = \frac{\alpha^2}{2\pi} + \frac{\norm{d}^2 k^2}{M^2 2\pi}).
\end{align*}

The remaining task is to distinguish the distribution $\nu_{\alpha} = N(0, V_0 = \frac{\alpha^2}{2\pi})$ from the above, by looking at the relative difference between their variances
\begin{align*}
    \frac{V_1 - V_0}{V_0} = \frac{\norm{d}^2 k^2}{M^2 \alpha^2}
\end{align*}
In the case when $\norm{d} = 1$, the ratio is the smallest. Proving it is non-negligible requires 
\begin{align*}
    \frac{k }{M \alpha} \ge \frac{1}{\text{poly}(n)}
\end{align*}

\end{proof}

\begin{remark}
Dean and Goldsmith's combined lower bounds on the constellation size and noise below
\begin{align*}
    &\text{Minimum Noise: } m\alpha/k^2 > \sqrt{n}\\
    &\text{Constellation Size: } M > m 2^{n \log \log n/\log n}
\end{align*}
force the ratio $k/M\alpha$ to be negligible, hence contradicting the corrected version of Lemma 2.

If the ratio is non-negligible, solution verification becomes possible using a sufficiently large polynomial number of samples. Independently, the attacks in \cite{steinfeld2015massive,sakzad2020comments,korzhik2017investigation} and the ILWE recovery algorithm \cite{bootle2018lwe} demonstrated to solve related non-modular noisy linear systems, under conditions involving sample size, noise magnitude and distribution of the public matrix. A further analysis should show that the parameter region required to repair Lemma 2 overlaps an efficiently recoverable region by the attacks and recovery algorithm. 
\end{remark}


\section{The gap in BDD-to-discrete-MIMO reduction}
\label{sec:bdd to mimo}

Granting the preceding reduction steps, Lemmas 1 and 3 of \cite{dean2017physical} are intended to allow the MIMO oracle to recover the secret from constructed discrete-MIMO samples whose effective noise has some unknown scale $\beta\leq\alpha$.\footnote{We grant Lemmas 4-6 for the present analysis and focus on the steps that must establish that the constructed samples follow the required MIMO distributions.} The central step in the classical part of the reduction is then to use this oracle to solve BDD. The resulting BDD oracle is subsequently used by the quantum stage of the Regev-style iterative reduction to generate discrete Gaussian samples with a smaller scale. Lemma 7 claims to establish this reduction from BDD to discrete-MIMO decoding. 

\paragraph{Lemma 7 \cite{dean2017physical}} MIMO-search$_{M,\alpha,k}$ to BDD$_{L,r}$. Let $\alpha > 0, k>0, m>0$, and $M>m2^{n \log \log n / \log n}$. Assume we have access to an oracle that, for all $\beta \le \alpha $, finds $\vc{x}$ given a polynomial number of samples from $\vc{A}_{M,\beta,k}$ (without knowing $\beta$). Then there exists an efficient algorithm that given an n-dimensional  lattice $L(\vc{A})$, a number $r > \sqrt{2}\eta(L(\vc{A}))$, and a target point $\vc{y}$ within distance $d < M\sigma \alpha / (k^2 r \sqrt{2})$ of $L(\vc{A})$, where $\sigma$ is the smallest eigenvalue of $\vc{A}^TA$, returns the unique $\vc{x} \in L(\vc{A})$  closest to $\vc{y}$ with overwhelming probability.

\paragraph{Proof strategy \cite{dean2017physical}} The proof of this lemma attempts to adapt Lemma 3.11 \cite{regev2009lattices} of the LWE reduction. It takes the BDD target $\vc{y}$ and a random dual lattice vector $\vc{v} \leftarrow D_{L^*, r}$ to produce the following sample
\begin{align*}
    (k\vc{v}/r, k \innerprod{\vc{v}, \vc{A}^{-1} \vc{y}}/rM + ke/r),
\end{align*}
and argues that it comes from the discrete-MIMO distribution over the domain $(k/r)L^* \times \R$, by considering the two parts of the sample separately. The first part clearly matches the discrete-MIMO sample's first part. By definition of BDD$_{L,d}$, $\vc{y} = A\vc{c} + \delta$ where $||\delta||<d$, the second part is then rewritten as 
\begin{align*}
    k \innerprod{\vc{v}, \vc{A}^{-1} \vc{y}}/rM + ke/r =  \underbrace{\innerprod{k\vc{v}/r,  \vc{c}}/M}_{\text{embedded secret}} + \underbrace{(\innerprod{k\vc{v}/r,  \vc{A}^{-1} \delta}/M + ke/r)}_{\text{aggregated noise}} \\
\end{align*}
The proof then continues the discussion of why this also follows the desired discrete-MIMO distribution. 

\subsection{Distributional mismatch}

\begin{proposition}
\label{prop:prop1}
Let $L(\vc{A})$ be a full rank lattice, and let $\vc{y} = A \vc{c} + \delta$ be an arbitrary BDD$_{L,d}$ target, where $\vc{c} \in \Z^n$ and $0 < d < \lambda_1(L)/2$. Sample $\vc{v} \leftarrow D_{L^*,r}$ and $\eta_0 \leftarrow \nu_{k\alpha/r}$ independently, define $\vc{a} = k\vc{v}/r$, then construct 
\begin{align*}
    b = \innerprod{\vc{a}, \vc{A}^{-1}\vc{y}}/M + \eta_0.
\end{align*}
If either $\vc{c} \notin \Omega^n$ or $\delta \neq 0$, then there do not exist a secret $\vc{s} \in \Omega^n$ and a parameter $0 < \beta \le \alpha$ such that $(\vc{a}, b)$ has the discrete-MIMO distribution $D_{M, \beta, k}$ obtained by sampling $\vc{a} \leftarrow D_{(k/r)L^*, k}$ and $e \leftarrow \nu_{\beta}$ independently, and outputting 
\begin{align*}
    \left(\vc{a}, \innerprod{\vc{a}, \vc{s}}/M + e \right).
\end{align*}
\end{proposition}



\begin{proof}
Let $\vc{z}=\vc{A}^{-1}\vc{y}=\vc{c}+\vc{A}^{-1}\delta$ and $\tau=k\alpha/r$. Since $\vc{v}\leftarrow D_{L^*,r}$, the scaled sample $\vc{a}=k\vc{v}/r \sim D_{(k/r)L^*,k}$. 

Conditioned on any fixed $\vc{a}$ in the support
of this distribution, the constructed second component has
distribution
\begin{align}
\label{eq:lemma7 b given a}
    P_{b\mid\vc{a}}
    =
    N\left(
        \frac{\innerprod{\vc{a},\vc{z}}}{M},
        \frac{\tau^2}{2\pi}
    \right).
\end{align}

Suppose, toward a contradiction, that there exist
$\vc{s} \in \Omega^n$ and $0<\beta\leq\alpha$ such that $(\vc{a},b)$ has the claimed discrete-MIMO distribution $D_{M,\beta,k}$. Conditioned on the same $\vc{a}$, we must have
\begin{align}
\label{eq:mimo b given a}
    Q_{b\mid\vc{a}}
    =
    N\left(
        \frac{\innerprod{\vc{a},\vc{s}}}{M},
        \frac{\beta^2}{2\pi}
    \right).
\end{align}
For the two conditional Gaussian distributions to equal, it must satisfy for every $\vc{a}$ in the support of $D_{(k/r)L^*,k}$ that
\begin{align*}
    \beta=\tau \text{ and } \innerprod{\vc{a},\vc{z}-\vc{s}}=0
\end{align*}
Because $(k/r)L^*$ is full rank and spans $\R^n$, this implies
\begin{align*}
    \vc{z}=\vc{s}.
\end{align*}
Consequently,
\begin{align*}
    \vc{A}^{-1}\vc{y}=\vc{s} \in \Omega^n\subseteq\Z^n,
\end{align*}
and therefore $\vc{y}=A\vc{s} \in L(\vc{A})$ is itself a lattice vector.

Since $A\vc{c}$ is the unique closest lattice vector to
$\vc{y}$, it follows that
\begin{align*}
    A\vc{c}=\vc{y}=A\vc{s}.
\end{align*}
As $\vc{A}$ is nonsingular, this gives
\begin{align*}
    \vc{c}=\vc{s}\in\Omega^n \text{ and } \delta=\vc{0},
\end{align*}
contradicting the assumption that either
$\vc{c}\notin\Omega^n$ or $\delta\neq\vc{0}$.
\end{proof}

Below are the key remarks on why the proof strategy in \cite{dean2017physical} did not work.    

\begin{remark}[Missing conditional randomness]
\label{rmk:miss random}
In Regev's construction \cite{regev2009lattices}, the map
\begin{align*}
    \vc{v}
    &\longmapsto
    \vc{a}=L^{-1}\vc{v}\bmod p
\end{align*}
is many to one, hence revealing $\vc{a}$ does not determine $\vc{v}$ completely. Conditioned on a fixed public value $\vc{a}$, the vector $\vc{v}$ remains distributed according to a discrete Gaussian over the coset $pL+L\vc{a}$. In particular,
\begin{align*}
    \vc{v}\mid\vc{a}
    &\sim
    D_{pL+L\vc{a},r}.
\end{align*}
Consequently, the term $\innerprod{\vc{x}',\vc{v}}/p$ remains random after the public component $\vc{a}$ is known. Regev applies Corollary 3.10 \cite{regev2009lattices} to this conditional distribution and shows that this term, together with the independent Gaussian noise $e$, is statistically close to a centred Gaussian of some unknown scale $\beta\leq\alpha$. Because the argument applies to every possible public value $\vc{a}$, it establishes the required closeness of the joint distribution.

Dean and Goldsmith instead defines
\begin{align*}
    \vc{a}
    &=
    \frac{k}{r}\vc{v}.
\end{align*}
For $k,r>0$, this map is injective, hence revealing $\vc{a}$ determines $\vc{v}$ completely. For a fixed BDD target, the additional term
\begin{align*}
    \innerprod{\vc{a},\vc{q}}
    +
        \innerprod{\vc{a},\vc{A}^{-1}\vc{\delta}}/M
\end{align*}
is deterministic conditioned on $\vc{a}$. The conditional distribution of the second part of the constructed sample $(\vc{a},b)$ is therefore a Gaussian whose
centre depends on the public component, rather than an independent $0$-centred Gaussian.
\end{remark}

\begin{remark}[Missing moduli]
\label{rmk:miss mod}
In Lemma 3.10 of LWE reduction \cite{regev2009lattices}, the constructed sample $(\vc{a}, b)$ has the form
\begin{align}
\label{eq:rmk-eq2}
a &= L^{-1} \vc{v} \bmod p,\nonumber \\
b &= \innerprod{\vc{x}, \vc{v}}/p + e \bmod 1 = \innerprod{\kappa_{L^*}(\vc{x}), \vc{v}}/p + (\underbrace{\innerprod{\vc{x}'/p, \vc{v}}}_{\text{random}} + e)  \bmod 1
\end{align}
and 
\begin{align}
\label{eq:rmk-eq1}
    \innerprod{\kappa_{L^*}(\vc{x}), \vc{v}} \bmod p = \innerprod{\vc{s}, \vc{a}} \bmod p
\end{align}
so dividing both sides of \Cref{eq:rmk-eq1} by $p$ then reducing modulo 1 removes the integer term, leaving only
\begin{align*}
    \innerprod{\kappa_{L^*}(\vc{x}), \vc{v}}/p \bmod 1 &= \innerprod{\vc{s}, \vc{a}}/p \bmod 1,
\end{align*}
which is the desired term for LWE distribution samples. Furthermore, as discussed in the previous remark, conditioning on $\vc{a}$ makes $\vc{v}$ follow a discrete Gaussian $D_{pL+L\vc{a},r}$ over the shifted coset, hence when adding with the noise $\vc{e}$ in \Cref{eq:rmk-eq2} it becomes a continuous Gaussian noise with larger scale.

In \cite{dean2017physical}'s Lemma 7 construction, if $\vc{c} \notin \Omega^n$ then $\vc{c} = \vc{s} + M\vc{q}$ for some $\vc{q} \in \Z^n$. The constructed sample $(\vc{a},b)$ has the forms
\begin{align}
\label{eq:rmk-eq3}
    a &= k\vc{v}/r, \nonumber \\
    b &= \innerprod{\vc{a}, \vc{s}}/M + (\underbrace{\innerprod{\vc{a},\vc{q}}+ \innerprod{\vc{a}, \vc{A}^{-1} \vc{\delta}}/M}_{\text{non-integer, fixed}} + \eta_0).
\end{align}
Since $\vc{a} \sim D_{(k/r)L^*, k}$ has a non-integral lattice support, \Cref{eq:rmk-eq3} can be rewritten as  
\begin{align*}
    b &= \innerprod{\vc{a},\vc{s}}/M + l(a) + \theta(\vc{a}) + e, \text{ where }\\
    l(\vc{a}) &= \floor{\innerprod{\vc{a},\vc{q}}+ \innerprod{\vc{a}, \vc{A}^{-1} \vc{\delta}}/M} \text{ and } \theta(a) \in [0,1).
\end{align*}
Again, by \Cref{rmk:miss random} $l(\vc{a})+\theta(\vc{a})$ is deterministic conditioned on $\vc{a}$. Even adding the missing modular 1, the term $\theta(\vc{a})$ is still a constant term, which will shift the noise $e$ to mismatch the discrete-MIMO noise distribution.

\end{remark}

\begin{remark}[Statistical detectability of the mean mismatch]
\label{rmk:stats dist}
\Cref{prop:prop1} proves that the two distributions $P_{b|\vc{a}}$ and $Q_{b|\vc{a}}$ are not identical. It is still possible, however, that they are within negligible statistical distance, so that the discrete-MIMO oracle can be utilised. 

To analyse their statistical distance, we first note that $P_{b | \vc{a}}$ and $Q_{b|\vc{a}}$ can have the same variance. 
In the  proof of Lemma 7 \cite{dean2017physical}, the constructed sample $(\vc{a},b)$ has an independent noise term $\eta_0=ke/r$, where $e\leftarrow\nu_\alpha$. For Dean and Goldsmith to apply Lemma 5 \cite{dean2017physical}, the noise scale must satisfy $k\alpha/r \le\alpha$. Hence, it is possible to have $\beta = k\alpha/r$. Denote $\sigma_0=\beta/\sqrt{2\pi}$, the distributions then become
\begin{align*}
    P_{b\mid\vc{a}}
    =
    N\left(
        \frac{\innerprod{\vc{a},\vc{z}}}{M},
        \sigma_0^2
    \right) \text{ and }
    Q_{b\mid\vc{a}}
    =
    N\left(
        \frac{\innerprod{\vc{a},\vc{s}}}{M},
        \sigma_0^2
    \right).
\end{align*}
The statistical distance between two equal variance Gaussian distributions is
\begin{align*}
    \Delta
    \left(
        N(\mu_1,\sigma^2),
        N(\mu_2,\sigma^2)
    \right)
    =
    4\Phi\left(
        \frac{|\mu_1-\mu_2|}{2\sigma}
    \right)-2,
\end{align*}
where $\Phi$ is the cumulative distribution function of $N(0,1)$.

Since the constructed sample and the discrete-MIMO sample have the same marginal distribution on $\vc{a}$, the statistical distance of their joint distributions is
\begin{align*}
    \Delta(P_{(\vc{a},b)},Q_{(\vc{a},b)})
    =
    \mathbb{E}_{\vc{a}}
    \left[
        4\Phi\left(
            \frac{
                \sqrt{2\pi}
                |\innerprod{\vc{a},\vc{z}-\vc{s}}|
            }{
                2M\beta_0
            }
        \right)-2
    \right].
\end{align*}
Recall $\vc{z} = \vc{c} + \vc{A}^{-1} \vc{\delta} = \vc{s} + M \vc{q} + \vc{A}^{-1} \vc{\delta}$. Hence, a large $M$ may suppress the BDD offset term $\innerprod{\vc{a},\vc{A}^{-1} \vc{\delta}}/M$, but it does not suppress the coefficient reduction term $\innerprod{\vc{a},\vc{q}}$. The statistical distance above may therefore be negligible for some restricted BDD targets where $\innerprod{\vc{a},\vc{q}}$ is small, but Dean and Goldsmith provides no argument showing that it is negligible for the arbitrary targets required by Lemma 7. 
\end{remark}

\subsection{Potential repairs to Lemma 7}

We briefly consider several possible modifications to the proof of Lemma 7 in \cite{dean2017physical}. None appears to constitute a local repair. Each either leaves part of the distributional mismatch unresolved or changes the source or target problem of the claimed reduction.

The first possibility is to restrict the BDD coefficient vector in
\begin{align*}
    \vc{y}
    =
    A\vc{c}+\vc{\delta}
\end{align*}
to $\vc{c}\in\Omega^n$. This removes the coefficient-reduction term because one may take $\vc{s}=\vc{c}$ and $\vc{q}=0$. Dean and Goldsmith's response nevertheless remains
\begin{align*}
    b
    =
    \frac{\innerprod{\vc{a},\vc{s}}}{M}
    +
    \frac{
        \innerprod{\vc{a},\vc{A}^{-1}\vc{\delta}}
    }{M}
    +
    \eta_0.
\end{align*}
Hence, the BDD offset still produces a mean displacement that is $\vc{a}$ dependent.

Moreover, the restriction $\vc{c}\in\Omega^n$ changes the source problem. Instead of receiving an arbitrary BDD target near $L(\vc{A})$, the reduction would receive a target from the restricted set
\begin{align*}
    \left\{
        A\vc{c}+\vc{\delta}
        \mid
        \vc{c}\in\Omega^n,\ 
        \|\vc{\delta}\|_2<d
    \right\}.
\end{align*}
This is a bounded coefficient BDD problem near the finite subset $\vc{A}\Omega^n\subset L(\vc{A})$. Standard BDD allows the closest lattice vector to have an arbitrary coefficient vector in $\Z^n$. Consequently, a separate reduction from standard BDD to this variant would be required. This restriction also does not produce the BDD $\bmod M$ problem analogous to the BDD $\bmod p$ problem used in Regev's reduction \cite{regev2009lattices}.

A second possibility is to reveal only a modular image of the public component, for example
\begin{align*}
    \overline{\vc{a}}
    =
    \vc{a}\bmod M\Z^n,
\end{align*}
rather than the complete vector $\vc{a}$. This would make the public map potentially many to one, but it does not reproduce the conditional coset structure used by Regev.

In particular, two lifts $\vc{a}$ and $\vc{a}+M\vc{t}$ have the same modular image, but their secret terms satisfy
\begin{align*}
    \frac{
        \innerprod{\vc{a}+M\vc{t},\vc{s}}
    }{M}
    =
    \frac{\innerprod{\vc{a},\vc{s}}}{M}
    +
    \innerprod{\vc{t},\vc{s}}.
\end{align*}
Because Dean and Goldsmith's response is not reduced modulo $1$, the additional integer remains in the response. 

Defining a new discrete-MIMO distribution whose public component is $\overline{\vc{a}}$ would change the discrete-MIMO problem and require a new analogue of Lemma 1 connecting the modified discrete problem to the original continuous MIMO problem.

A third possibility is to replace the response by
\begin{align*}
    \overline{b}
    =
    \left(
        \frac{\innerprod{\vc{a},\vc{s}}}{M}
        +
        \mu_{\vc{s}}(\vc{a})
        +
        \eta_0
    \right)
    \bmod 1.
\end{align*}
As shown in \Cref{rmk:miss mod}, this modification replaces the unwanted displacement by
\begin{align*}
    \mu_{\vc{s}}(\vc{a})\bmod 1,
\end{align*}
which is not guaranteed to vanish. Moreover, if the complete vector $\vc{a}$ remains public, then the injectivity discussed in \Cref{rmk:miss random} remains unchanged. Conditioned on $\vc{a}$, the residual modular displacement is still deterministic and the wrapped noise remains centred at an $\vc{a}$ dependent value.

Combining a modular $1$ response with a carefully chosen quotient of the public component could potentially restore both integral cancellation and conditional randomness. Such a construction, however, would be a wrapped or periodic problem closer to LWE \cite{regev2009lattices} or continuous LWE \cite{bruna2021continuous}, rather than Dean and Goldsmith's original non-modular MIMO problem.

Based on the above, repairing Lemma 7 requires more than correcting an isolated parameter or adding a modular operation. A successful replacement must simultaneously map arbitrary BDD coefficients into the MIMO constellation, preserve sufficient randomness after conditioning on the public component, and produce independent centred noise of an allowed scale. None of the modifications considered above currently satisfies all three requirements.


\section{Beyond \citeauthor{dean2017physical}'s construction}

Motivated by the same MIMO idea, the subsequent work of \cite{liu2024physical} introduced the MMIMO-Precoding$_{S,\sigma_e,J}$ problem using massive-MIMO channels and precoding. Its construction provides the legitimate receiver with a secret ``good'' lattice basis, while the eavesdropper is given only a corresponding ``bad'' public basis. The legitimate receiver is therefore intended to decode in polynomial time, whereas the eavesdropper is claimed to face an exponentially hard CVP instance. This trapdoor asymmetry is intended to address the decoding-correctness problem in \cite{dean2017physical} identified by \cite{sakzad2020comments}. However, the claimed reduction from BDD to MMIMO-Precoding in Theorem~2 of \cite{liu2024physical} suffers from the same problems as Dean and Goldsmith's Lemma 7, as discussed in \Cref{sec:bdd to mimo}.

\paragraph{Theorem 2 \cite{liu2024physical}} [MMIMO-Precoding$_{S,\sigma_e,\vc{J}}$ to BDD$_{L(\vc{A}),\sqrt{N}\sigma_e}$ problem] Let $S \ge m2^{N\log\log N /\log N}$, $\sigma_e$ is the standard deviation of the noise, $x \in [0,S)$, $\vc{J}$ is the proposed precoding matrix, $N$ is the antennas number. Suppose there is an efficient algorithm that can solve MMIMO-Precoding$_{S,\sigma_e,\vc{J}}$ problem. Then there exists an efficient algorithm that, given an $N$-dimensional lattice $L(\vc{A})$, can solve the BDD$_{L(\vc{A}),\sqrt{N}\sigma_e}$ problem. Therefore, because BDD$_{L(\vc{A}),\sqrt{N}\sigma_e}$ problem is assumed to be hard, we can conclude that MMIMO-Precoding$_{S,\sigma_e,\vc{J}}$ problem is hard.

The key problematic step of their intended reduction proof is the same construction as that in \cite{dean2017physical}, with the second term being deterministic, hence shifting the Gaussian noise  distribution to non-zero centred.
\begin{align*}
    \frac{k}{r}\innerprod{\vc{v}, B^{-1} \vc{y}} + e' &= \left\langle \frac{k}{r}\vc{v}, B^{-1}(B X + \vc{\delta}) \right\rangle + e'\\
    &= \left\langle \frac{k}{r}\vc{v}, X \right\rangle + \underbrace{\left\langle \frac{k}{r}\vc{v}, B^{-1}\vc{\delta} \right\rangle}_{\text{deterministic}} + e'.
\end{align*}


\section{Conclusion}

We examined the revised Dean-Goldsmith reduction from lattice problems to MIMO decoding and showed that its published proof does not establish the claimed hardness result. Several steps fail to justify that the constructed samples follow the distributions required by the MIMO oracle. These problems arise more broadly from the absence of the modular cancellation and conditional randomness used in Regev's LWE reduction. Our analysis complements the existing decoding attacks by identifying proof-level failures in the reduction and raises similar concerns for subsequent constructions employing the same strategy. These findings do not rule out physical layer security for MIMO systems, but show that it is not established by the existing lattice reduction.


\section*{Acknowledgements}
The author thanks Dr. Yuxuan Li for her initial investigation of this problem. 
This work was partially supported by Deakin University, SEBE PRESS Funding Scheme 2024.

\newpage
\bibliography{references}
\bibliographystyle{plainnat}

\end{document}